\documentclass[10pt,twocolumn, nofootinbib]{revtex4-2}

\usepackage{assumptionsofphysics}
\usepackage{tikz}
\usepackage{breakurl}

\def\>{\rangle}
\def\<{\langle}
\DeclareMathOperator{\erf}{erf}

\begin{document}

\title{The unphysicality of Hilbert spaces}
\author{Gabriele Carcassi}
\affiliation{Physics Department, University of Michigan, Ann Arbor, MI 48109}
\author{Francisco Calder\'on}
\affiliation{Philosophy Department, University of Michigan, Ann Arbor, MI 48109}
\author{Christine A. Aidala}
\affiliation{Physics Department, University of Michigan, Ann Arbor, MI 48109}

\date{\today}

\begin{abstract}
We argue that Hilbert spaces are not suitable to represent quantum states mathematically, in the sense that they require properties that are untenable by physical entities. We first demonstrate that the requirements posited by complex inner product spaces are physically justified. We then show that completeness in the infinite-dimensional case requires the inclusion of states with infinite expectations, coordinate transformations that take finite expectations to infinite ones and vice versa, and time evolutions that transform finite expectations to infinite ones in finite time. This means we should be wary of using Hilbert spaces to represent quantum states as they turn a potential infinity into an actual infinity. The main point of the paper, then, is to raise awareness that results that rely on the completeness of Hilbert spaces may not be physically significant. While we do not claim to know what a physically more appropriate closure should be, we note that Schwartz spaces, among other things, guarantee that the expectations of all polynomials of position and momentum are finite, their elements are uniquely identified by these expectations, and they are closed under Fourier transform.  

\end{abstract}

\maketitle

\section{Introduction}

Hilbert spaces have been the cornerstone of quantum mechanics since von Neumann's 1932 book~\cite{von_neumann_mathematische_1996}. While the formalism is mathematically consistent and most accept it without question, it does not have a clear physical justification~\cite{heathcote_1990, hardy_2001}. In fact, von Neumann himself expressed his dissatisfaction not long after publication of his book, and spent considerable effort looking for alternatives~\cite{vonNeumannHilbert_1996}.

In an effort to find a precise physical justification, we realized that no such justification can exist as infinite-dimensional Hilbert spaces always contain elements that cannot have physical counterparts. To put it simply, they must include states that have infinite expectation values for position, momentum, energy and so on. This means that a truly general statement in quantum theory must be valid also on these irrelevant unphysical cases, with the consequence that either physical results are left out, or they are obscured in a complicated patchwork of theorems with different domains of validity. We therefore suspect that some of the mathematical difficulties encountered when trying to give a more precise mathematical foundation to quantum theory stem from this fundamental mismatch, and we argue more generally that mathematical spaces used in physics should already come equipped with the proper structure that excludes physically pathological behavior.

In this work we briefly present our findings, leaving more mathematical detail to the appendix. We will see that complex inner product spaces are indeed needed, and should be considered physical. We will see that the problem lies in requiring completeness, which should be understood as an unphysical mathematical requirement. We will show that other potential options exist, offering Schwartz spaces as an example. These guarantee the expectation values of all polynomials of position and momentum to be finite, and they are also closed under Fourier transform.

\section{On physicality}

When we first drafted this paper, we naively thought that the advantage of having a well-posed mathematical structure that represents, and only represents, physically well-defined objects would be self-evident. It does not seem to be the case: a reviewer from a previous journal submission stated that ``the question of whether the limiting states are actually physical is rather pointless;'' conversely, others see no problem in working with quantum field theory path integrals that are not mathematically well-defined. We are therefore in the absurd, from our perspective, position of having to defend that, yes, the math we need to use in physics has to be mathematically precise and, yes, the math we need to use in physics has to represent---perhaps idealized---physically realizable objects. The cost of not requiring a clear and tight connection between the math and the physics is all around us: an ever-increasing number of interpretations of quantum mechanics, abstract mathematical work whose connection to physics is unclear, theoretical physics more and more disconnected from experimental physics. The reader who finds this obvious can skip this section.

One problem is that most physicists see mathematics as just a tool for calculation, much like oscilloscopes or telescopes are tools to gather data. As long as they do what they are supposed to do, the details are irrelevant because the physics lies somewhere else. If they do not work correctly, it is just a detail for the expert to fix. But mathematics plays a much more fundamental role in physics: it is the language in which we specify our models. When we say, ``States are rays in a Hilbert space,'' we are providing a specification for a physical system.  We are saying that there is a particular physical system which under certain circumstances can be mathematically modeled in a particular way.  So, if Hilbert spaces do not work, it is not simply a math problem: the model is wrong, which is a physics problem.  This is the question we are examining in this paper: can the math we use in quantum mechanics model a particular class of physical systems correctly?  Is our physical problem well specified?

To make this determination, one needs a standard to judge the physicality of a formal definition~\cite{redei_tension_2020, north_physics_2021} and, rather tellingly, physics tradition does not provide one. Here we describe the standard that we follow in our work~\cite{aop-book}.  Suppose we are studying a physical system. A mathematical definition is \textbf{physical} if it properly characterizes the physical system in the given conditions. That is, if we are able to \textbf{justify on physical grounds} that the given mathematical definition is needed to capture and only capture a particular aspect of the physical system under the given conditions. On the other hand, a mathematical definition is \textbf{unphysical} if it can be shown to require properties or operations that cannot have a physical counterpart. The question of this paper, then, is whether the use of Hilbert spaces to characterize quantum systems can be justified on physical grounds, from physical requirements. This is the task at hand.

While it seems obvious to us, it may be worth stating that \textbf{it is not up to mathematicians to justify the use of a particular mathematical structure in physics}. Ultimately, the justification cannot be solely in terms of mathematical matters, but rather it has to rely on that elusive ``physical intuition.'' In fact, the task is exactly to make that elusive ``physical intuition'' precise enough to turn it into mathematical definitions. This is the job of a physicist. In our specific case, the core argument we will use is that states described by measurable quantities with infinite or ill-defined expectation values are not physically realizable. To us, and to every experimental physicist we talked to, ``It doesn't make sense experimentally,'' and this is the best argument possible to show that something is unphysical. ``It can't be done, not even in principle.'' 

If the problem of justifying the use of mathematical structures in physics is, ultimately, a physical problem, it also follows that it cannot be solved simply by a different mathematical approach. For example, rigged Hilbert spaces were developed to give a mathematical foundation to Dirac’s bra-ket formalism~\cite{Dirac1947}, and to be able to talk about eigenstates of position and momentum, working mathematically with both continuous and discrete bases on an interchangeable level~\cite{Madrid_2005,Celeghini2016}. The question is not whether these can be useful at a mathematical level to solve certain problems. They clearly are. The question at hand is whether, on physical grounds, this is a judicious idea. Should we really treat continuous and discrete spectra the same, or should we consider that they are, on physical grounds, two rather different things? Should the fact that the energy spectrum is typically discrete for bound states and continuous for unbound states make us wonder if there is something in the differing topologies of the spectra that is physically important? Should we then make the difference more prominent in our thinking, rather than trying to construct mathematical tools that treat them the same way?  Similarly, one can take the algebraic route and start with an algebra of observables. This, some argue, makes some mathematical problems easier to solve, which indeed it does. It still begs the question: why is it that, on physical grounds, observables should form a $C^*$-algebra? What, then, is the suitable physically meaningful state space? One of the ``successes'' of the algebraic approach was to construct a Hilbert space from an abstract $C^*$-algebra via the so-called GNS construction. Thus in both these frameworks, rigged Hilbert spaces and the $C^*$-algebra approach, Hilbert spaces still play a crucial role. If Hilbert spaces are physically problematic, these approaches are still relying on a mathematical structure that is physically problematic.  Looking to a completely different field, within the hierarchy of safety controls, elimination, substitution, and engineering controls are strongly preferred over administrative controls to manage hazards---putting up a fence at the edge of a cliff is much safer than simply putting a sign saying, ``Beware of cliff.''  By using Hilbert spaces for physical calculations in quantum mechanics, the community is effectively relying on ``administrative controls'' to avoid incorrect results, i.e.~considering if a mathematically allowed calculation makes physical sense and discarding it if not.  For physics problems in realms close to everyday human experience, our physical intuition is very reliable.  However, for problems within e.g.~quantum physics, quantum field theory, or cosmology, our physical intuition is much less powerful, and in fact these areas of physics rely much more heavily on the mathematics.  Our paper is effectively advocating that we try to implement ``engineering controls'' to ensure the soundness of calculations in quantum mechanics---if we can find a suitable mathematical space that does not allow unphysical objects in the first place, then there is no risk of accidentally failing to discard unphysical results. 

While the argument is ultimately physical, and it is therefore the job of a physicist to make these arguments, it has to be a physicist that is well-versed in foundational issues in mathematics. The justification, in fact, will need precise understanding of the mathematical definitions and their implications for both the physics and the math. It is often the case that the mathematical detail informs the physics and should not be swept under the rug.  Let us give an example. In a previous work some of us found that a set of physically distinguishable objects must be a $T_0$ second-countable topological space where each open set represents a verifiable statement, a statement allowing a test that terminates in finite time if and only if the statement is true~\cite{aop-book}. This result tells us why functions in physics must be ``well-behaved'' (i.e.~topologically continuous in the right topology), why sets with cardinality greater than the continuum are unphysical (i.e.~they cannot be given a $T_0$ second-countable topology) and why the Banach-Tarski paradox does not apply in physics (i.e.~non-Borel sets are unphysical as they cannot be associated with an experimental test). This tells us that, at the very least, non-separable Hilbert spaces are unphysical as they have cardinality greater than the continuum~\cite{Earman2020}. This ties together ideas from point-set topology, logic, set theory, measure theory and computer science. Clearly, one needs to have at least a basic understanding of these fields to be able to follow the argument, and most degree programs in physics do not typically discuss these topics in depth.\footnote{A previous reviewer conflated the completeness of the Hilbert space with the existence of a complete orthonormal basis of said space.} However, in the end, only a physicist can judge whether the evidence-based requirement of physics justifies the mathematical requirements.

To summarize, we are not interested here in understanding whether Hilbert spaces have nice mathematical features for calculating or proving results. Moreover, the fact that other tools have been developed to solve other mathematical problems is not relevant to this discussion. We are only interested in the question of mathematical representation of quantum states, which are physical entities we prepare in a lab. States in physics are the output of a preparation procedure, a characterization that works for both ``pure states'' of mechanics as well as ensembles in statistical mechanics. They are what we produce, manipulate and, in the end, measure. What physical requirements are we capturing when we model them mathematically? Are Hilbert spaces capturing, and only capturing, physically tenable assumptions, or not?  

\section{On the physicality of Hilbert spaces}

As we have discussed in the previous section, to argue for the physicality of a mathematical definition we need to make a justification \textbf{on physical grounds} that the definition captures a physical property.  We will see that the inner product space part of the Hilbert space definition is physical as it is required to to compute entropy and probabilities. We have also seen that to argue for the unphysicality of a mathematical definition we need to show that the definition implies objects that are physically unrealizable. We will see that the completeness part of the Hilbert space definition is unphysical as it forces us to include states with infinite (or undefined) expectation values.

\subsection{Inner product space}

Recall that a Hilbert space is a complete inner product space. That is, it is a vector space with an inner product, and its Cauchy sequences converge under the norm defined by the inner product. We will look at the first two structures only briefly, leaving all detail in the appendix, as their physical significance is unsurprising.

Roughly speaking, a vector space is a set whose elements can be combined linearly: they can be multiplied by constants and summed. If one regards superpositions as physical requirements of quantum theory, one may think this structure is physical. However, the ability to interpret linear combinations as superpositions requires the ability to keep track of orthonormal vectors. If we can only know whether two vectors are in the same ray, we cannot recover the idea of probability amplitudes, as those are the coefficients of a linear combination between normalized and orthogonal states that return another normalized state. Therefore, the vector space structure, by itself, is not physical.

If we add the inner product, we can define norm and orthogonality, and therefore we can indeed express superpositions. Additionally, the inner product structure defines the Born rule, the transition probability during measurement. While some may consider these physical requirements of quantum mechanics, their status is subject to interpretation~\cite{albert_quantum_1994, wallace_everett_2013, howard_complementarity_2021, ghirardi_unified_1986}. Some consider superpositions as mathematical constructs; the Born rule brings in the notion of measurements. We therefore reframe those requirements in terms of properties of ensembles. Here we provide a summary, and leave the details to the appendix.

Since we can physically prepare statistical ensembles, a physical theory must be able to characterize them. Mathematically, these form a convex set, where the convex combinations are the statistical mixtures.\footnote{There is a rich literature on generalized probabilistic theories (GPTs) in which people reconstruct quantum theory by imposing requirements on convex spaces~\cite{gpt_overview_2021}}. In quantum mechanics, the geometry of the space of statistical ensembles is exactly defined by the complex inner product and vector space structures, which means properties of the inner product space can be reexpressed as properties of the space of mixtures. For example, we can characterize a superposition in the following way: a pure state $\psi$ is expressible as a superposition of other pure states $\{\phi_i\}$ if and only if there is a mixed state $\rho$ that can be equivalently expressed as a mixture of $\{\phi_i\}$ or as a mixture of $\psi$ and other pure states. That is, the ability to create a superposition is conceptually equivalent to the ability to prepare the same mixture in different ways, another feature particular to quantum mechanics. Additionally, the norm of the inner product (i.e. the probabilities from the Born rule), together with relative phases, can be shown to be equivalent to determining the entropy of all equal mixtures of pairs of states. These requirements are necessary in the sense that both classical and quantum mechanics need to keep track of how ensembles combine under mixing and what entropy each ensemble has. As such, this approach avoids the issues of interpretation.

The argument can be summed up as follows. On physical grounds, for any physical theory, we must be able to prepare statistical mixtures and characterize their entropy. The complex inner product space structure is exactly required to model the correct relationships of quantum mixtures, therefore the mathematical definition is necessitated by physical requirements. \textbf{The inner product space structure is physical}.

\subsection{Completeness}

We now turn to the last property of Hilbert spaces: completeness. Mathematically, this means that the space includes the limit of every Cauchy sequence; physically, this property is more difficult to characterize. We can look at an equivalent characterization~\cite[Theorem 13.8]{roman_2008}. An inner product space is complete if, given a series of vectors
$$ \sum _{k=0}^{\infty }u_{k}$$
such that
$$ \sum _{k=0}^{\infty }\|u_{k}\|<\infty,$$
the series converges in the space. That is, if the series of the norm converges, the series converges. This characterization of completeness is in terms of linear combinations. It implies that we can prepare states that are superpositions of infinitely many elements or, equivalently, as we have seen, mixed states comprised of infinitely many pure states. Is this a physical requirement or just mathematical convenience?

Experimentally, we can create a superposition by splitting a beam along multiple paths, which is typical in quantum optics. The above mathematical expression means splitting the original state into infinitely many progressively smaller components, and being able to place them where we please. However, it also assumes that we can place them over an infinite spatial range: we are literally saying that we can prepare states that are spread over infinite regions. This does not sound physical.

Again, given the difficulties of understanding superposition, let us reframe the result in terms of statistical quantities. Suppose we can prepare pure states with arbitrary finite expectation values for position. Then, completeness means we can also prepare pure states that fall off at infinity as $\frac{1}{x^2}$. The expectation value of $x$ would be proportional to the integral of a function that falls off as $\frac{1}{x}$, which diverges. Clearly, this does not make physical sense. When we posit that a quantity, like position, has an infinite range of possible values, we simply mean that the observed value can be arbitrarily large, not that we can literally observe an infinite value.\footnote{This is the difference between potential infinity and actual infinity.}

The argument can be summed up as follows. On physical grounds, we assume the existence of quantities, like position, that can take arbitrarily large values. This requires that, in principle, for each possible value, there must exist some state whose expectation matches that value. Completeness forces us to include states with an infinite expectation value, which cannot be implemented physically. Thus the mathematical definition implicitly assumes objects that cannot be physically realized, so \textbf{the use of completeness for the space of quantum systems is unphysical.}\footnote{As rigged Hilbert spaces are Hilbert spaces with additional structure, they cannot solve this problem. For further discussion, see the appendix.}

\subsection{Impact of unphysicality}

It should not be surprising that closure over infinite sequences is problematic, as properly handling infinity is one of the key problems in mathematics. We should understand how completeness tries to address that problem and its impact on the physics.

If an inner product space is finite-dimensional, it is automatically complete.\footnote{This is exactly why completeness is mathematically desirable: it makes infinite-dimensional spaces behave like finite-dimensional ones.} If the space is not complete, we can always complete it by enlarging it. The prevalent attitude among physicists is that this does not pose a problem: we have a bigger set of objects than we strictly need, and we simply discard what we do not need. Unfortunately, having a bigger space can introduce new problems.

For example, suppose we want to define volumes in three-dimensional Euclidean space $\mathbb{R}^3$. Formally, we want to find a measure $\mu(U)$ that returns the volume of a region $U$. Naively, we could take a region to be any subset of $\mathbb{R}^3$. However, this leads to the Banach–Tarski paradox~\cite{tao2021introduction}, which states that we can take a unit ball, cut it into five pieces, and reassemble those pieces into two balls each of the same volume as the original. Clearly, this does not make physical sense. However, we get a consistent definition if we restrict regions to be Borel sets, instead of arbitrary subsets. Therefore, bigger is not always better in mathematics:\footnote{In another work~\cite{aop-book} some of us show that the Borel sets are exactly those sets that can be associated with experimental procedures. We have often found that physical requirements naturally map to well-behaved spaces.} we need to look more closely at what the math dragged in.

Unitary transformations play an important role in the theory of Hilbert spaces for at least two reasons. The first is that two Hilbert spaces that differ by a unitary transformation are mathematically equivalent. For example, changes of basis and changes of variables are unitary transformations. This is already a problem for physics: the space of wave functions for $n$ particles is the Hilbert space $L^2(\mathbb{R}^{3n})$. For any value of $n$, these infinite-dimensional spaces are separable, meaning that they all have a countable basis, and so are unitarily equivalent. That is, the state space for one particle is the same as the state space for, say, a billion particles. We can therefore theoretically construct a quantum channel that encodes the information of a billion particles into a single one. Note that this is not possible in classical mechanics: the state space for $n$ classical particles is $\mathbb{R}^{6n}$. For different $n$, these spaces are not even topologically equivalent.

The second reason for the importance of unitary transformations is that they represent possible time evolution for an isolated system. The issue is that unitary transformations in a Hilbert space can map states with finite expectation values to those with infinite ones, leading to obvious problems.

As a concrete example, consider the following two wave functions
\begin{align}
\psi(x) &= \sqrt{\frac{e^{-x^2}}{\sqrt{\pi}}} \\
\phi(x) &= \sqrt{\frac{1}{\pi(x^2 + 1)}}.
\end{align}
Both states are symmetric and, therefore, will have zero expectation value for $x$. The probability distribution for the first is a Gaussian; therefore, all moments of the distribution, the expectation for $x^n$, are finite. The one for the second goes to infinity like $\frac{1}{x^2}$, so the expectation will diverge for all moments above the first. We can now find a change of variable $y=y(x)$ that transforms one distribution into the other. We need to set
\begin{equation}
\begin{aligned}
\int_{0}^{y(x)} \phi^\dagger(\hat{y}) \phi(\hat{y}) d\hat{y} &= \int_{0}^{x} \psi^\dagger(\hat{x}) \psi(\hat{x}) d\hat{x} \\
\int_{0}^{y(x)} \frac{1}{\pi(\hat{y}^2 + 1)} d\hat{y} &= \int_{0}^{x} \frac{e^{-\hat{x}^2}}{\sqrt{\pi}} d\hat{x} \\
\frac{\tan^{-1}(y(x))}{\pi} &= \frac{\erf(x)}{2} \\
y(x) &= \tan \left(\frac{\pi}{2}\erf(x)\right). \\
\end{aligned}
\end{equation}
Therefore, through a variable change, we can change a state with a finite expectation value for all moments greater than the first into a state with an infinite expectation value for the same moments. Variable changes are unitary transformations on the Hilbert space. Therefore, the theory of Hilbert spaces, as applied to quantum mechanics, would tell us that the two representations are equivalent.

On physical grounds, these are clearly not equivalent. Reconstructing a state from its expectation values is a common technique, quantum tomography~\cite{banaszek2013focus}. If a unitary transformation can change finite expectations to infinite ones, we effectively change what can be distinguished experimentally. Infinite uncertainty in one frame would map to finite uncertainty in another. Distributions that are distinguishable in one frame would not be distinguishable in another: in the example above, if we increased the variance, one observer would see all the even moments change, but the other would see them all equally infinite. The math, as it stands, takes physically inequivalent states as equivalent, making it impossible later to separate physical questions from unphysical ones.

Moreover, we can create the following family of transformations
\begin{equation}
z(x, t) = \cos(\omega t) x + \sin(\omega t) \tan \left(\frac{\pi}{2}\erf(x)\right).
\end{equation}
This is clearly continuous in $t$ and, for $t=0$, we have the identity transformation, and therefore, mathematically, there is no reason why this would not be a valid time evolution. What this does is stretch back and forth the wave function, such that the state keeps oscillating between the two functions above, making the expectation value oscillate from finite to infinite in finite time. This, again, is physically untenable.

Again, these issues often do not arise when performing actual calculations because physics makes us work in the corner of the structure where things are reasonable. The issue is that when one tries to prove theorems or general results on quantum theory, one is currently forced to consider the entire Hilbert space. It is unreasonable to expect all theorems valid in all physical cases to have valid generalizations for unphysical ones. We are likely missing out on useful results by not considering what is provable within the subset of physical cases. We, therefore, wonder whether many of the technical issues we have in quantum theory (e.g.~the non-existence of a valid measure for path integrals~\cite{glimm_quantum_1987}, inequivalent representations for interacting theories~\cite{earman_haags_2006}, the non-existence of the categorical tensor product for Hilbert spaces~\cite{garrett_tensor_2010, sorkin_inside_2022}, etc.)  are simply a consequence of using unphysical objects at a fundamental level.

\section{On the possible physicality of Schwartz spaces}

While we do need infinite-dimensional spaces to handle unbound quantities, such as the number of particles, or continuous quantities, such as position, it seems that completeness is not the proper characterization. While finding the proper one goes beyond the scope of this work, we want to show that fruitful lines of inquiry do exist.

A reasonable alternative would be to require the expectation of all polynomials of position and momentum to be finite. In quantum mechanics, this means requiring all expectations of $\frac{1}{2}(X^nP^m + P^mX^n)$ to be finite. This restricts our space to those functions that are infinitely smooth (i.e.~all derivatives exist) and decrease very rapidly as $x$ goes to infinity: these are the Schwartz functions.

The Schwartz space has several other interesting properties~\cite{moretti_spectral_2017, reed_methods_1980, hall_quantum_2013, rauch_solution_1991}, explored in more detail in the appendix. For example:

\begin{itemize}
\item its elements are fully identified by the expectations - two Schwartz functions are the same if and only if all expectations are the same; quantum tomography on all polynomials fully reconstructs the state
\item it is an inner product space with the same norm as $L^2$
\item it is a dense subset of $L^2$ - any element of the Hilbert space can be approximated by a Schwartz function with an arbitrary level of precision
\item it is complete with respect to the expectation values - a sequence of functions for which all expectation values of all polynomials of position and momentum converge will converge in the Schwartz space
\item the Schwartz spaces over $\mathbb{R}^{n}$ are topologically different if $n$ is different
\item its closure over all Cauchy sequences recovers the Hilbert space - the only thing missing are the Cauchy sequences for which the limit of at least one expectation is not well defined
\item it is closed under Fourier transform - the Fourier transform of a Schwartz function is a Schwartz function
\item it plays a fundamental role in the theory of distributions - objects like the $\delta$-function are mathematically defined on top of the Schwartz space
\item the Hermite polynomials, the solutions of the harmonic oscillator, are Schwartz functions
\end{itemize}
These properties make the Schwartz space a much more reasonable candidate to capture the physics~\cite{albert2022bosonic}. The convergence of expectation values of actual observables replaces convergence over vectors in an abstract space.\footnote{Compare this statement with the ``weak convergence'' of von Neumann algebras.} However, it does not seem to us that this is a final and general answer. It is not clear to us how Schwartz spaces are generalized in the case where the topology of the underlying space is not $\mathbb{R}^n$. Moreover, something more general will be needed in the case of indistinguishable particles of variable number, which is handled by the Fock space in the standard treatment. However, we do think that the properties of Schwartz spaces may provide guidance for a possible comprehensive solution.\footnote{Note that $C^*$-algebras model observables and therefore do not solve the problem at hand of modeling states.}

\section{Conclusion}

We have made the case that Hilbert spaces in quantum mechanics are \textbf{unphysical} as they necessarily bring in states with infinite expectation values, variable changes that turn finite expectations into infinite ones, and time evolution operators that do the same in finite time. This flawed characterization of infinity is the likely cause of outstanding technical problems in the foundations of quantum mechanics and quantum field theory, and makes it impossible to have a fully meaningful understanding of the theory.

The more general problem is that the physics community has become complacent in simply accepting mathematical definitions without properly understanding their limit of applicability to physical theories. We believe that meaningful progress on the foundations of physics cannot happen without understanding the implicit assumptions embedded in the most basic mathematical tools.

\section*{Acknowledgments}
The authors wish to thank David Carf\`{i} and Noel Swanson for helpful discussions.  This paper is part of the ongoing Assumptions of Physics project~\cite{aop-book}, which aims to identify a handful of physical principles from which the basic laws can be rigorously derived.

\section*{Data Availability Statement}
Not applicable.

\section*{Conflict of interest}
On behalf of all authors, the corresponding author states that there is no conflict of interest. 

\bibliography{bibliography}

\begin{thebibliography}{33}%
\makeatletter
\providecommand \@ifxundefined [1]{%
 \@ifx{#1\undefined}
}%
\providecommand \@ifnum [1]{%
 \ifnum #1\expandafter \@firstoftwo
 \else \expandafter \@secondoftwo
 \fi
}%
\providecommand \@ifx [1]{%
 \ifx #1\expandafter \@firstoftwo
 \else \expandafter \@secondoftwo
 \fi
}%
\providecommand \natexlab [1]{#1}%
\providecommand \enquote  [1]{``#1''}%
\providecommand \bibnamefont  [1]{#1}%
\providecommand \bibfnamefont [1]{#1}%
\providecommand \citenamefont [1]{#1}%
\providecommand \href@noop [0]{\@secondoftwo}%
\providecommand \href [0]{\begingroup \@sanitize@url \@href}%
\providecommand \@href[1]{\@@startlink{#1}\@@href}%
\providecommand \@@href[1]{\endgroup#1\@@endlink}%
\providecommand \@sanitize@url [0]{\catcode `\\12\catcode `\$12\catcode
  `\&12\catcode `\#12\catcode `\^12\catcode `\_12\catcode `\%12\relax}%
\providecommand \@@startlink[1]{}%
\providecommand \@@endlink[0]{}%
\providecommand \url  [0]{\begingroup\@sanitize@url \@url }%
\providecommand \@url [1]{\endgroup\@href {#1}{\urlprefix }}%
\providecommand \urlprefix  [0]{URL }%
\providecommand \Eprint [0]{\href }%
\providecommand \doibase [0]{https://doi.org/}%
\providecommand \selectlanguage [0]{\@gobble}%
\providecommand \bibinfo  [0]{\@secondoftwo}%
\providecommand \bibfield  [0]{\@secondoftwo}%
\providecommand \translation [1]{[#1]}%
\providecommand \BibitemOpen [0]{}%
\providecommand \bibitemStop [0]{}%
\providecommand \bibitemNoStop [0]{.\EOS\space}%
\providecommand \EOS [0]{\spacefactor3000\relax}%
\providecommand \BibitemShut  [1]{\csname bibitem#1\endcsname}%
\let\auto@bib@innerbib\@empty
\bibitem [{\citenamefont {von Neumann}(1932)}]{von_neumann_mathematische_1996}%
  \BibitemOpen
  \bibfield  {author} {\bibinfo {author} {\bibfnamefont {J.}~\bibnamefont {von
  Neumann}},\ }\href {https://doi.org/10.1007/978-3-642-61409-5} {\emph
  {\bibinfo {title} {Mathematische {Grundlagen} der {Quantenmechanik}}}},\
  \bibinfo {edition} {2nd}\ ed.\ (\bibinfo  {publisher} {Springer-Verlag},\
  \bibinfo {year} {1996/1932})\BibitemShut {NoStop}%
\bibitem [{\citenamefont {Heathcote}(1990)}]{heathcote_1990}%
  \BibitemOpen
  \bibfield  {author} {\bibinfo {author} {\bibfnamefont {A.}~\bibnamefont
  {Heathcote}},\ }\bibfield  {title} {\bibinfo {title} {Unbounded operators and
  the incompleteness of quantum mechanics},\ }\href
  {https://doi.org/10.1086/289572} {\bibfield  {journal} {\bibinfo  {journal}
  {Philosophy of Science}\ }\textbf {\bibinfo {volume} {57}},\ \bibinfo {pages}
  {523–534} (\bibinfo {year} {1990})}\BibitemShut {NoStop}%
\bibitem [{\citenamefont {Hardy}(2001)}]{hardy_2001}%
  \BibitemOpen
  \bibfield  {author} {\bibinfo {author} {\bibfnamefont {L.}~\bibnamefont
  {Hardy}},\ }\bibfield  {title} {\bibinfo {title} {Quantum theory from five
  reasonable axioms}\ }\href {https://doi.org/10.48550/arXiv.quant-ph/0101012}
  {10.48550/arXiv.quant-ph/0101012} (\bibinfo {year} {2001})\BibitemShut
  {NoStop}%
\bibitem [{\citenamefont {Rédei}(1996)}]{vonNeumannHilbert_1996}%
  \BibitemOpen
  \bibfield  {author} {\bibinfo {author} {\bibfnamefont {M.}~\bibnamefont
  {Rédei}},\ }\bibfield  {title} {\bibinfo {title} {Why {J}ohn von {N}eumann
  did not like the {H}ilbert space formalism of quantum mechanics (and what he
  liked instead)},\ }\href
  {https://doi.org/https://doi.org/10.1016/S1355-2198(96)00017-2} {\bibfield
  {journal} {\bibinfo  {journal} {Studies in History and Philosophy of Science
  Part B: Studies in History and Philosophy of Modern Physics}\ }\textbf
  {\bibinfo {volume} {27}},\ \bibinfo {pages} {493} (\bibinfo {year}
  {1996})}\BibitemShut {NoStop}%
\bibitem [{\citenamefont {Rédei}(2020)}]{redei_tension_2020}%
  \BibitemOpen
  \bibfield  {author} {\bibinfo {author} {\bibfnamefont {M.}~\bibnamefont
  {Rédei}},\ }\bibfield  {title} {\bibinfo {title} {On the {Tension} {Between}
  {Physics} and {Mathematics}},\ }\href
  {https://doi.org/10.1007/s10838-019-09496-0} {\bibfield  {journal} {\bibinfo
  {journal} {Journal for General Philosophy of Science}\ }\textbf {\bibinfo
  {volume} {51}},\ \bibinfo {pages} {411} (\bibinfo {year} {2020})}\BibitemShut
  {NoStop}%
\bibitem [{\citenamefont {North}(2021)}]{north_physics_2021}%
  \BibitemOpen
  \bibfield  {author} {\bibinfo {author} {\bibfnamefont {J.}~\bibnamefont
  {North}},\ }\href {https://doi.org/10.1093/oso/9780192894106.001.0001} {\emph
  {\bibinfo {title} {Physics, {Structure}, and {Reality}}}}\ (\bibinfo
  {publisher} {Oxford University Press},\ \bibinfo {year} {2021})\BibitemShut
  {NoStop}%
\bibitem [{\citenamefont {Carcassi}\ and\ \citenamefont
  {Aidala}(2021)}]{aop-book}%
  \BibitemOpen
  \bibfield  {author} {\bibinfo {author} {\bibfnamefont {G.}~\bibnamefont
  {Carcassi}}\ and\ \bibinfo {author} {\bibfnamefont {C.~A.}\ \bibnamefont
  {Aidala}},\ }\href {https://doi.org/10.3998/mpub.12204707} {\emph {\bibinfo
  {title} {Assumptions of Physics}}}\ (\bibinfo  {publisher} {Michigan
  Publishing},\ \bibinfo {year} {2021})\BibitemShut {NoStop}%
\bibitem [{\citenamefont {Dirac}(1947)}]{Dirac1947}%
  \BibitemOpen
  \bibfield  {author} {\bibinfo {author} {\bibfnamefont {P.~A.~M.}\
  \bibnamefont {Dirac}},\ }\href@noop {} {\emph {\bibinfo {title} {The
  principles of {Quantum Mechanics}}}},\ \bibinfo {edition} {3rd}\ ed.\
  (\bibinfo  {publisher} {Clarendon Press},\ \bibinfo {address} {Oxford},\
  \bibinfo {year} {1947})\BibitemShut {NoStop}%
\bibitem [{\citenamefont {de~la Madrid}(2005)}]{Madrid_2005}%
  \BibitemOpen
  \bibfield  {author} {\bibinfo {author} {\bibfnamefont {R.}~\bibnamefont
  {de~la Madrid}},\ }\bibfield  {title} {\bibinfo {title} {The role of the
  rigged {H}ilbert space in quantum mechanics},\ }\href
  {https://doi.org/10.1088/0143-0807/26/2/008} {\bibfield  {journal} {\bibinfo
  {journal} {European Journal of Physics}\ }\textbf {\bibinfo {volume} {26}},\
  \bibinfo {pages} {287} (\bibinfo {year} {2005})}\BibitemShut {NoStop}%
\bibitem [{\citenamefont {E.~Celeghini}\ \emph {et~al.}(2016)\citenamefont
  {E.~Celeghini}, \citenamefont {Gadella},\ and\ \citenamefont {del
  Olmo}}]{Celeghini2016}%
  \BibitemOpen
  \bibfield  {author} {\bibinfo {author} {\bibfnamefont {E.}~\bibnamefont
  {E.~Celeghini}}, \bibinfo {author} {\bibfnamefont {M.}~\bibnamefont
  {Gadella}},\ and\ \bibinfo {author} {\bibfnamefont {M.~A.}\ \bibnamefont {del
  Olmo}},\ }\bibfield  {title} {\bibinfo {title} {Applications of rigged
  {Hilbert} spaces in quantum mechanics and signal processing},\ }\href
  {https://doi.org/10.1063/1.4958725} {\bibfield  {journal} {\bibinfo
  {journal} {Journal of Mathematical Physics}\ }\textbf {\bibinfo {volume}
  {57}},\ \bibinfo {pages} {072105} (\bibinfo {year} {2016})}\BibitemShut
  {NoStop}%
\bibitem [{\citenamefont {Earman}(2020)}]{Earman2020}%
  \BibitemOpen
  \bibfield  {author} {\bibinfo {author} {\bibfnamefont {J.}~\bibnamefont
  {Earman}},\ }\bibfield  {title} {\bibinfo {title} {Quantum physics in
  non-separable {H}ilbert spaces}\ }(\bibinfo {year} {2020})\ \bibinfo {note}
  {{A}vailable from \url{https://philsci-archive.pitt.edu/18363/}}\BibitemShut
  {NoStop}%
\bibitem [{\citenamefont {Albert}(1994)}]{albert_quantum_1994}%
  \BibitemOpen
  \bibfield  {author} {\bibinfo {author} {\bibfnamefont {D.~Z.}\ \bibnamefont
  {Albert}},\ }\href@noop {} {\emph {\bibinfo {title} {Quantum {Mechanics} and
  {Experience}}}}\ (\bibinfo  {publisher} {Harvard University Press},\ \bibinfo
  {address} {Cambridge, MA},\ \bibinfo {year} {1994})\BibitemShut {NoStop}%
\bibitem [{\citenamefont {Wallace}(2013)}]{wallace_everett_2013}%
  \BibitemOpen
  \bibfield  {author} {\bibinfo {author} {\bibfnamefont {D.}~\bibnamefont
  {Wallace}},\ }\bibfield  {title} {\bibinfo {title} {The {Everett}
  {Interpretation}},\ }in\ \href
  {https://doi.org/10.1093/oxfordhb/9780195392043.013.0014} {\emph {\bibinfo
  {booktitle} {The {Oxford} {Handbook} of {Philosophy} of {Physics}}}},\
  \bibinfo {editor} {edited by\ \bibinfo {editor} {\bibfnamefont
  {R.}~\bibnamefont {Batterman}}}\ (\bibinfo  {publisher} {Oxford University
  Press},\ \bibinfo {year} {2013})\ pp.\ \bibinfo {pages}
  {460--488}\BibitemShut {NoStop}%
\bibitem [{\citenamefont {Howard}(2021)}]{howard_complementarity_2021}%
  \BibitemOpen
  \bibfield  {author} {\bibinfo {author} {\bibfnamefont {D.}~\bibnamefont
  {Howard}},\ }\bibfield  {title} {\bibinfo {title} {Complementarity and
  {Decoherence}},\ }in\ \href {https://doi.org/10.1007/978-3-030-77367-0_8}
  {\emph {\bibinfo {booktitle} {Quantum {Arrangements}: {Contributions} in
  {Honor} of {Michael} {Horne}}}},\ \bibinfo {series and number} {Fundamental
  {Theories} of {Physics}},\ \bibinfo {editor} {edited by\ \bibinfo {editor}
  {\bibfnamefont {G.}~\bibnamefont {Jaeger}}, \bibinfo {editor} {\bibfnamefont
  {D.}~\bibnamefont {Simon}}, \bibinfo {editor} {\bibfnamefont {A.~V.}\
  \bibnamefont {Sergienko}}, \bibinfo {editor} {\bibfnamefont {D.}~\bibnamefont
  {Greenberger}},\ and\ \bibinfo {editor} {\bibfnamefont {A.}~\bibnamefont
  {Zeilinger}}}\ (\bibinfo  {publisher} {Springer International Publishing},\
  \bibinfo {year} {2021})\ pp.\ \bibinfo {pages} {151--175}\BibitemShut
  {NoStop}%
\bibitem [{\citenamefont {Ghirardi}\ \emph {et~al.}(1986)\citenamefont
  {Ghirardi}, \citenamefont {Rimini},\ and\ \citenamefont
  {Weber}}]{ghirardi_unified_1986}%
  \BibitemOpen
  \bibfield  {author} {\bibinfo {author} {\bibfnamefont {G.~C.}\ \bibnamefont
  {Ghirardi}}, \bibinfo {author} {\bibfnamefont {A.}~\bibnamefont {Rimini}},\
  and\ \bibinfo {author} {\bibfnamefont {T.}~\bibnamefont {Weber}},\ }\bibfield
   {title} {\bibinfo {title} {Unified dynamics for microscopic and macroscopic
  systems},\ }\href {https://doi.org/10.1103/PhysRevD.34.470} {\bibfield
  {journal} {\bibinfo  {journal} {Physical Review D}\ }\textbf {\bibinfo
  {volume} {34}},\ \bibinfo {pages} {470} (\bibinfo {year} {1986})},\ \bibinfo
  {note} {publisher: American Physical Society}\BibitemShut {NoStop}%
\bibitem [{\citenamefont {Plávala}(2023)}]{gpt_overview_2021}%
  \BibitemOpen
  \bibfield  {author} {\bibinfo {author} {\bibfnamefont {M.}~\bibnamefont
  {Plávala}},\ }\bibfield  {title} {\bibinfo {title} {General probabilistic
  theories: An introduction},\ }\href
  {https://doi.org/https://doi.org/10.1016/j.physrep.2023.09.001} {\bibfield
  {journal} {\bibinfo  {journal} {Physics Reports}\ }\textbf {\bibinfo {volume}
  {1033}},\ \bibinfo {pages} {1} (\bibinfo {year} {2023})}\BibitemShut
  {NoStop}%
\bibitem [{\citenamefont {Roman}(2008)}]{roman_2008}%
  \BibitemOpen
  \bibfield  {author} {\bibinfo {author} {\bibfnamefont {S.}~\bibnamefont
  {Roman}},\ }\href {https://doi.org/10.1007/978-0-387-72831-5} {\emph
  {\bibinfo {title} {Advanced Linear Algebra}}}\ (\bibinfo  {publisher}
  {Springer},\ \bibinfo {year} {2008})\BibitemShut {NoStop}%
\bibitem [{\citenamefont {Tao}(2021)}]{tao2021introduction}%
  \BibitemOpen
  \bibfield  {author} {\bibinfo {author} {\bibfnamefont {T.}~\bibnamefont
  {Tao}},\ }\href@noop {} {\emph {\bibinfo {title} {An Introduction to Measure
  Theory}}},\ Graduate Studies in Mathematics\ (\bibinfo  {publisher} {American
  Mathematical Society},\ \bibinfo {year} {2021})\BibitemShut {NoStop}%
\bibitem [{\citenamefont {Banaszek}\ \emph {et~al.}(2013)\citenamefont
  {Banaszek}, \citenamefont {Cramer},\ and\ \citenamefont
  {Gross}}]{banaszek2013focus}%
  \BibitemOpen
  \bibfield  {author} {\bibinfo {author} {\bibfnamefont {K.}~\bibnamefont
  {Banaszek}}, \bibinfo {author} {\bibfnamefont {M.}~\bibnamefont {Cramer}},\
  and\ \bibinfo {author} {\bibfnamefont {D.}~\bibnamefont {Gross}},\ }\bibfield
   {title} {\bibinfo {title} {Focus on quantum tomography},\ }\href@noop {}
  {\bibfield  {journal} {\bibinfo  {journal} {New Journal of Physics}\ }\textbf
  {\bibinfo {volume} {15}},\ \bibinfo {pages} {125020} (\bibinfo {year}
  {2013})}\BibitemShut {NoStop}%
\bibitem [{\citenamefont {Glimm}\ and\ \citenamefont
  {Jaffe}(1987)}]{glimm_quantum_1987}%
  \BibitemOpen
  \bibfield  {author} {\bibinfo {author} {\bibfnamefont {J.}~\bibnamefont
  {Glimm}}\ and\ \bibinfo {author} {\bibfnamefont {A.}~\bibnamefont {Jaffe}},\
  }\href {https://doi.org/10.1007/978-1-4612-4728-9} {\emph {\bibinfo {title}
  {Quantum {Physics}: {A} {Functional} {Integral} {Point} of {View}}}},\
  \bibinfo {edition} {2nd}\ ed.\ (\bibinfo  {publisher} {Springer-Verlag},\
  \bibinfo {address} {New York},\ \bibinfo {year} {1987})\BibitemShut {NoStop}%
\bibitem [{\citenamefont {Earman}\ and\ \citenamefont
  {Fraser}(2006)}]{earman_haags_2006}%
  \BibitemOpen
  \bibfield  {author} {\bibinfo {author} {\bibfnamefont {J.}~\bibnamefont
  {Earman}}\ and\ \bibinfo {author} {\bibfnamefont {D.}~\bibnamefont
  {Fraser}},\ }\bibfield  {title} {\bibinfo {title} {Haag’s {Theorem} and its
  {Implications} for the {Foundations} of {Quantum} {Field} {Theory}},\ }\href
  {https://doi.org/10.1007/s10670-005-5814-y} {\bibfield  {journal} {\bibinfo
  {journal} {Erkenntnis}\ }\textbf {\bibinfo {volume} {64}},\ \bibinfo {pages}
  {305} (\bibinfo {year} {2006})}\BibitemShut {NoStop}%
\bibitem [{\citenamefont {Garrett}()}]{garrett_tensor_2010}%
  \BibitemOpen
  \bibfield  {author} {\bibinfo {author} {\bibfnamefont {P.}~\bibnamefont
  {Garrett}},\ }\bibfield  {title} {\bibinfo {title} {Non-existence of tensor
  products of {H}ilbert spaces},\ }\href@noop {} {\ }\bibinfo {note} {Available
  from
  \url{https://www-users.cse.umn.edu/~garrett/m/v/nonexistence_tensors.pdf}}\BibitemShut
  {NoStop}%
\bibitem [{\citenamefont {Sorkin}(2022)}]{sorkin_inside_2022}%
  \BibitemOpen
  \bibfield  {author} {\bibinfo {author} {\bibfnamefont {R.~D.}\ \bibnamefont
  {Sorkin}},\ }\bibfield  {title} {\bibinfo {title} {An {Inside} {View} of the
  {Tensor} {Product}},\ }in\ \href {https://doi.org/10.1142/9789811270437_0022}
  {\emph {\bibinfo {booktitle} {Particles, {Fields}, and {Topology}:
  {Celebrating} {A}. {P}. {Balachandran}}}}\ (\bibinfo  {publisher} {World
  Scientific},\ \bibinfo {year} {2022})\ pp.\ \bibinfo {pages}
  {235--263}\BibitemShut {NoStop}%
\bibitem [{\citenamefont {Moretti}(2017)}]{moretti_spectral_2017}%
  \BibitemOpen
  \bibfield  {author} {\bibinfo {author} {\bibfnamefont {V.}~\bibnamefont
  {Moretti}},\ }\href {https://doi.org/10.1007/978-3-319-70706-8} {\emph
  {\bibinfo {title} {Spectral {Theory} and {Quantum} {Mechanics}:
  {Mathematical} {Foundations} of {Quantum} {Theories}, {Symmetries} and
  {Introduction} to the {Algebraic} {Formulation}}}},\ \bibinfo {edition}
  {2nd}\ ed.\ (\bibinfo  {publisher} {Springer},\ \bibinfo {year}
  {2017})\BibitemShut {NoStop}%
\bibitem [{\citenamefont {Reed}\ and\ \citenamefont
  {Simon}(1980)}]{reed_methods_1980}%
  \BibitemOpen
  \bibfield  {author} {\bibinfo {author} {\bibfnamefont {M.}~\bibnamefont
  {Reed}}\ and\ \bibinfo {author} {\bibfnamefont {B.}~\bibnamefont {Simon}},\
  }\href@noop {} {\emph {\bibinfo {title} {Methods of {Modern} {Mathematical}
  {Physics}}}},\ Vol.\ \bibinfo {volume} {1: Functional Analysis}\ (\bibinfo
  {publisher} {Academic Press},\ \bibinfo {year} {1980})\BibitemShut {NoStop}%
\bibitem [{\citenamefont {Hall}(2013)}]{hall_quantum_2013}%
  \BibitemOpen
  \bibfield  {author} {\bibinfo {author} {\bibfnamefont {B.~C.}\ \bibnamefont
  {Hall}},\ }\href {https://doi.org/10.1007/978-1-4614-7116-5} {\emph {\bibinfo
  {title} {Quantum {Theory} for {Mathematicians}}}},\ \bibinfo {series}
  {Graduate {Texts} in {Mathematics}}, Vol.\ \bibinfo {volume} {267}\ (\bibinfo
   {publisher} {Springer New York},\ \bibinfo {year} {2013})\BibitemShut
  {NoStop}%
\bibitem [{\citenamefont {Rauch}(1991)}]{rauch_solution_1991}%
  \BibitemOpen
  \bibfield  {author} {\bibinfo {author} {\bibfnamefont {J.}~\bibnamefont
  {Rauch}},\ }\bibfield  {title} {\bibinfo {title} {Solution of {Initial}
  {Value} {Problems} by {Fourier} {Synthesis}},\ }in\ \href
  {https://doi.org/10.1007/978-1-4612-0953-9_3} {\emph {\bibinfo {booktitle}
  {Partial {Differential} {Equations}}}},\ \bibinfo {series and number}
  {Graduate {Texts} in {Mathematics}},\ \bibinfo {editor} {edited by\ \bibinfo
  {editor} {\bibfnamefont {J.}~\bibnamefont {Rauch}}}\ (\bibinfo  {publisher}
  {Springer},\ \bibinfo {address} {New York, NY},\ \bibinfo {year} {1991})\
  pp.\ \bibinfo {pages} {95--136}\BibitemShut {NoStop}%
\bibitem [{\citenamefont {Albert}(2022)}]{albert2022bosonic}%
  \BibitemOpen
  \bibfield  {author} {\bibinfo {author} {\bibfnamefont {V.~V.}\ \bibnamefont
  {Albert}},\ }\href@noop {} {\bibinfo {title} {Bosonic coding: introduction
  and use cases}} (\bibinfo {year} {2022}),\ \Eprint
  {https://arxiv.org/abs/2211.05714} {arXiv:2211.05714 [quant-ph]} \BibitemShut
  {NoStop}%
\bibitem [{\citenamefont {Nielsen}\ and\ \citenamefont
  {Chuang}(2010)}]{nielsen_chuang_2010}%
  \BibitemOpen
  \bibfield  {author} {\bibinfo {author} {\bibfnamefont {M.~A.}\ \bibnamefont
  {Nielsen}}\ and\ \bibinfo {author} {\bibfnamefont {I.~L.}\ \bibnamefont
  {Chuang}},\ }\href {https://doi.org/10.1017/CBO9780511976667} {\emph
  {\bibinfo {title} {Quantum Computation and Quantum Information: 10th
  Anniversary Edition}}}\ (\bibinfo  {publisher} {Cambridge University Press},\
  \bibinfo {year} {2010})\BibitemShut {NoStop}%
\bibitem [{\citenamefont {Landsman}(2017)}]{Landsman2017}%
  \BibitemOpen
  \bibfield  {author} {\bibinfo {author} {\bibfnamefont {K.}~\bibnamefont
  {Landsman}},\ }\bibinfo {title} {Symmetry in quantum mechanics},\ in\ \href
  {https://doi.org/10.1007/978-3-319-51777-3_5} {\emph {\bibinfo {booktitle}
  {Foundations of Quantum Theory: From Classical Concepts to Operator
  Algebras}}}\ (\bibinfo  {publisher} {Springer International Publishing},\
  \bibinfo {year} {2017})\ pp.\ \bibinfo {pages} {125--188}\BibitemShut
  {NoStop}%
\bibitem [{\citenamefont {Treves}(1967)}]{fourierCompact}%
  \BibitemOpen
  \bibfield  {author} {\bibinfo {author} {\bibfnamefont {F.}~\bibnamefont
  {Treves}},\ }\bibfield  {title} {\bibinfo {title} {Fourier transforms of
  distributions with compact support.}\ }(\bibinfo  {publisher} {Elsevier},\
  \bibinfo {year} {1967})\ pp.\ \bibinfo {pages} {305--313}\BibitemShut
  {NoStop}%
\bibitem [{\citenamefont {Markushevich}(1966)}]{markushevich2014entire}%
  \BibitemOpen
  \bibfield  {author} {\bibinfo {author} {\bibfnamefont {A.~I.}\ \bibnamefont
  {Markushevich}},\ }\href {https://doi.org/10.1016/C2013-0-12422-1} {\emph
  {\bibinfo {title} {Entire functions}}}\ (\bibinfo  {publisher} {Elsevier},\
  \bibinfo {year} {1966})\BibitemShut {NoStop}%
\bibitem [{\citenamefont {Schmüdgen}(2017)}]{moment_problem_2017}%
  \BibitemOpen
  \bibfield  {author} {\bibinfo {author} {\bibfnamefont {K.}~\bibnamefont
  {Schmüdgen}},\ }\href {https://doi.org/10.1007/978-3-319-64546-9} {\emph
  {\bibinfo {title} {The Moment Problem}}}\ (\bibinfo  {publisher} {Springer},\
  \bibinfo {year} {2017})\BibitemShut {NoStop}%
\end{thebibliography}%

\newcommand{\pj}[1] {\underbar{$#1$}}

\section*{Appendix}

This appendix presents the same arguments in the article's main body, but expanded with full mathematical detail. The figure and the derivations in the Appendix are courtesy of the Assumptions of Physics project~\cite{aop-book}.

\subsection{Physicality of inner product space structure}

The first task is to show that the vector space structure is not, by itself, physical. Since quantum states are represented by rays of a Hilbert space, and rays can be defined simply on top of the vector space structure alone, one may think that the vector space structure, by itself, is physical. The issue is that, without an inner product, one cannot make physical sense of the coefficients one uses for linear combinations, and, therefore, superpositions become ill-defined.

To show the problem, we prove the following proposition that tells us that there are infinitely many ways to decompose a vector representing a state in terms of vectors representing two other states, meaning that the coefficients of the linear combination are arbitrary.
\begin{prop}[The vector space structure is not enough to uniquely identify superpositions]\label{vector_insufficient} Let $\mathcal{H}$ be a complex vector space and let $\textbf{P}(\mathcal{H})$ be its projective space. Let $\psi, \phi_1, \phi_2 \in \mathcal{H}$ such that $\psi = c_1 \phi_1 + c_2 \phi_2$ for some $c_1, c_2 \in \mathbb{C}$. Let $\pj{\psi}, \pj{\phi_1}, \pj{\phi_2} \in \textbf{P}(\mathcal{H})$ be the corresponding rays in the projective space. Then for any $\hat{c}_1, \hat{c}_2 \in \mathbb{C}$ we can find  $\hat{\phi}_1, \hat{\phi}_2 \in \mathcal{H}$ such that $\pj{\hat{\phi}}_1 = \pj{\phi_1}$, $\pj{\hat{\phi}}_2 = \pj{\phi_2}$ and
$$\psi = \hat{c}_1 \hat{\phi}_1 + \hat{c}_2 \hat{\phi}_2.$$
\end{prop}

\begin{proof}
Let $\hat{\phi}_1 = \frac{c_1}{\hat{c}_1} \phi_1$. Since they differ only by a constant, they are elements of the same ray. That is, $\pj{\hat{\phi}}_1 = \pj{\phi_1}$. Similarly, we let $\hat{\phi}_2 = \frac{c_2}{\hat{c}_2} \phi_2$ and we have $\pj{\hat{\phi}}_2 = \pj{\phi_2}$. We have
\begin{equation}
\begin{aligned}
\hat{c}_1 \hat{\phi}_1 + \hat{c}_2 \hat{\phi}_2 &= \hat{c}_1 \frac{c_1}{\hat{c}_1} \phi_1 + \hat{c}_2 \frac{c_2}{\hat{c}_2} \phi_2 \\
&= c_1 \phi_1 + c_2 \phi_2 \\
&= \psi.
\end{aligned}
\end{equation}
\end{proof}

The intuition that a state has a single linear decomposition in terms of other vectors requires, first of all, that the vectors are normalized. To identify the coefficients as probability amplitudes, that is that $\sum_i |c_i|^2 = 1$, further requires decomposition on orthogonal vectors.\footnote{Mathematically, this means turning the vectors $\{\phi_i\}$ into an orthonormal set using Gram-Schmidt orthonormalization, which requires the inner product.} That is, the inner product is required for linear combinations to be understood physically as superpositions. Therefore, the vector space structure by itself is insufficient to provide a clear connection to physics.

We could now use the above reasoning to argue that the inner product space structure, as a whole, is physical because it indeed characterizes both superpositions and the Born rule. These are features required by quantum mechanics and, therefore, representationally necessary and physically motivated. Many will find this argument convincing enough, but there are several problems.

First of all, it is not clear whether superpositions are actual physical entities or whether they are just mathematical constructs. This is an area where different interpretations may differ~\cite{albert_quantum_1994, wallace_everett_2013, howard_complementarity_2021}, and the physicality of a mathematical representation should not be interpretation-dependent. Second, in its standard characterization, the Born rule brings in one of the most controversial aspects of the theory, measurements, and thus similarly suffers from an interpretation problem. Lastly, the inner product is critical for both the Born rule, which is a transition probability between states before and after a measurement, and linear decomposition into superposition, which is done at equal time. Given that we need to show that the mathematical definition captures and only captures a particular aspect of a physical system, we need to formulate a physical requirement that is interpretation-independent and grounded in experimental practice.

We will show that the mathematical tools required to characterize statistical mixtures can be reformulated in terms of superpositions and the Born rule. Mathematically, these will be in terms of the space of density operators instead of the vector space itself. Since preparation of statistical mixtures is well grounded in experimental practice, their characterization provides a well posed physical requirement.

The idea is to link superposition to another feature of quantum mechanics: the non-unique decomposition of mixed states into pure states. For example, if we take the maximally mixed state for a spin 1/2 system, this can be implemented with an equal mixture of spin up and spin down or with an equal mixture of spin left and spin right. The idea is that a pure state $\psi$ is expressible as a superposition of other pure states $\{\phi_i\}$ exactly because there is a mixed state $\rho$ that can be equivalently expressed as a mixture of $\{\phi_i\}$ or as a mixture of $\psi$ and other pure states. Subspaces of a Hilbert space group together pure states that can provide equivalent statistical mixtures. To keep the reasoning as general as possible, we will not require $\{\phi_i\}$ to be orthogonal to each other.

\begin{prop}\label{prop_densitySpan}
Let $\mathcal{H}$ be a complex inner product space. Let $\{\psi_i\}, \{\phi_j\} \in \mathcal{H}$ be two finite sets of normalized vectors. Suppose there exists a density operator that can be expressed as a strictly convex combination of either set. That is, we can find two sequences $\{a_i\}, \{b_j\} \in \mathbb{R}_{>0}$ such that
\begin{enumerate}
\item $\sum_i a_i = 1$
\item $\sum_j b_j = 1$
\item $\sum_i a_i |\psi_i\>\<\psi_i| = \sum_j b_j |\phi_j\>\<\phi_j|$,
\end{enumerate}
then the span of the two sets is the same.
\end{prop}

\begin{proof}
Consider $\rho = \sum_i a_i |\psi_i\>\<\psi_i|$. For a normalized vector $\varphi \in \mathcal{H}$ we have
$$p(\varphi|\rho)=tr(\rho |\varphi\>\<\varphi|) = \sum_i a_i |\<\psi_i|\varphi\>|^2$$
Since all $a_i$ are positive, $p(\varphi|\rho) = 0$ if and only if $\varphi$ is orthogonal to all $\psi_i$.

Now suppose $\rho = \sum_i a_i |\psi_i\>\<\psi_i| = \sum_j b_j |\phi_j\>\<\phi_j|$. Then a vector $\varphi$ is orthogonal to all $\psi_i$ if and only if it is also orthogonal to all $\phi_i$, which means $span(\{\psi_i\}) = span(\{\phi_j\})$) as they have the same orthogonal complement.
\end{proof}

\begin{prop}\label{prop_decomposition}
Let $\mathcal{H}$ be a complex inner product space. Let $\psi \in \mathcal{H}$ be a normalized vector and $\rho : \mathcal{H} \to \mathcal{H}$ a density operator of finite rank such that $\<\psi|\rho|\psi\>\neq 0$. Then we can express $\rho$ as a mixture of pure states that includes $|\psi\>\<\psi|$. That is, there exists a finite set of normalized vectors $\{\psi_i\} \in \mathcal{H}$ such that $\psi_1 = \psi$ and $\rho = \sum_i p_i |\psi_i\>\<\psi_i|$, where $\{p_i\} \in \mathbb{R}_{>0}$ and $\sum_i p_i = 1$.
\end{prop}

\begin{proof}
The set of all possible density operators is a convex set, which can be characterized as a subset of a vector space $V$. Let $\rho$ be a density operator such that $\<\psi|\rho|\psi\>\neq 0$. This will be a point in the convex set. If $\rho = |\psi\>\<\psi|$, we are done. If not, the two elements $\rho$ and $|\psi\>\<\psi|$ will identify a line in $V$, which must intersect the boundary of the convex set in two places. One must be $\psi$ since $\psi$ is a pure state and, therefore, an extreme point. Let $\rho_1$ be the other point. If it is an extreme point, we are done as $\rho$ is a convex combination of $\psi$ and $\rho_1$. If not, $\rho_1$ is on a part of the boundary of the convex set that is not strictly convex. That region will be a convex set constrained on a subspace $V_1$ of $V$. Note that $V_1$ does not include $\psi$, and it does not span the direction identified by $\rho_1$ and $\psi$. This means that we can't have $\rho_1 = \rho$ as $\<\psi|\rho|\psi\>\neq 0$. Therefore we have that $\rho$ is a strictly convex combination of $\psi$ and $\rho_1$, which must be a density operator of finite rank lower than $\rho$. We can express $\rho_1$ as a convex combination of finitely many extreme points $\psi_2, \psi_3, ..., \psi_n$, each of them corresponding to a pure state. We thus have that $\rho$ is a strictly convex combination of pure states that includes $\psi$.
\end{proof}

\begin{prop}[Superposition is non-unique decomposition of mixed states.]\label{prop_superpositionIsDecomposition}
Let $\mathcal{H}$ be a complex vector space. Let $\psi \in \mathcal{H}$ be a normalized vector and $\{\phi_j\} \in \mathcal{H}$ a finite set of normalized vectors that does not contain $\psi$. Then $\psi$ is a linear combination of $\{\phi_j\}$ if and only if there exists a density operator $\rho$ that can be expressed both as a strictly convex combination of $\{\phi_j\}$, and of another set of normalized vectors that contain $\psi$. That is, we can find a finite set of normalized vectors $\{\psi_i\} \in \mathcal{H}$ and two sequences $\{a_i\}, \{b_j\} \in \mathbb{R}_{>0}$ such that
\begin{enumerate}
\item $\psi_1 = \psi$
\item $\sum_i a_i = 1$
\item $\sum_j b_j = 1$
\item $\sum_i a_i |\psi_i\>\<\psi_i| = \sum_j b_j |\phi_j\>\<\phi_j|$,
\end{enumerate}
\end{prop}

\begin{proof}
Suppose $\psi$ is a normalized vector that can be written as $\psi = \sum_{i} c_i \phi_i$ where $\{\phi_i\} \in \mathcal{H}$ is a finite set of normalized vectors that does not contain $\psi$ and $\{c_i\} \in \mathbb{C}$. Since $\{\phi_i\}$ are not necessarily orthogonal, it is not necessarily true that $|\<\psi|\phi_i\>| = |c_i|^2$ or that $|\<\psi|\phi_i\>| \neq 0$ for all $i$. However, since $\psi$ is normalized and $\psi \neq \phi_i$ for all $i$, there must be at least two vectors in $\{\phi_i\}$ such that $|\<\psi|\phi_i\>| \neq 0$. We can then find $\{p_i\} \in \mathbb{R}_{>0}$ for which $\sum_i p_i =1$, such that $\rho = \sum_i p_i |\phi_i\>\<\phi_i|$ and $\<\psi|\rho|\psi\> \neq 0$. By Proposition \ref{prop_decomposition}, $\rho$ can be expressed as a strictly convex combination of a set of normalized vectors that includes $\psi$.

Conversely, suppose we have a density operator $\rho$ that can be expressed as a strictly convex combination of $\{\phi_i\}$ and another set of normalized vectors containing $\psi$. By Proposition \ref{prop_densitySpan} the span of $\{\phi_i\}$ will include $\psi$ and therefore $\psi$ is a linear combination of $\{\phi_i\}$.
\end{proof}

As for the inner product, we show that the Born rule can be understood as characterizing the entropy of equal mixtures over state pairs. Note that the Born rule does not depend on the absolute phase, but it does depend on relative phases. If one expands a vector using superposition, in fact, the relative phase will introduce interference terms. Therefore recovering the Born rule means recovering the inner product up to arbitrary absolute phases, which do not contain any physical content.

Associating the Born rule with entropy gives us a geometric understanding of the inner product at equal time instead of as a probability of transition, which brings it more in line with the use during superposition. It should also be noted that the von Neumann entropy of a mixture can be understood as the minimum Shannon entropy over all possible measurement contexts~\cite{nielsen_chuang_2010}, which again shows the connection between the geometry of the convex space of statistical ensembles and the inner product.

\begin{prop}[Inner product is an entropic structure]\label{prop_innerProductIsEntropy}
Let $\mathcal{H}$ be a complex inner product space. Let $I : \mathcal{H} \times \mathcal{H} \to \mathbb{R}$ be the von Neumann entropy of the equal mixture of two pure states. That is
\begin{equation}
I(\psi, \phi) = S\left(\frac{1}{2}|\psi\>\<\psi| + \frac{1}{2}|\phi\>\<\phi|\right).
\end{equation}
Let $p : \mathcal{H} \times \mathcal{H} \to \mathbb{R}$ be the square of the inner product. That is
\begin{equation}
p(\psi, \phi) = |\<\psi| \phi\>|^2.
\end{equation}
Then $p$ and $I$ can be reconstructed from each other. That is, there exists an invertible function $f$ such that
\begin{equation}
I(\psi, \phi) = f(p(\psi, \phi)).
\end{equation}
\end{prop}

\begin{proof}
Let $\psi, \phi \in \mathcal{H}$ be two normalized vectors. They will identify a two-dimensional subspace which can be thought, without loss of generality, as a qubit and therefore can be represented by a Bloch sphere. The below picture represents the intersection of the Bloch sphere with the plane identified by $\psi$ and $\phi$. Let $\rho = \frac{1}{2}|\psi\>\<\psi| + \frac{1}{2}|\phi\>\<\phi|$ be the density operator representing an equal mixture of the two states. It will be represented by the midpoint of the chord connecting $\psi$ and $\phi$. Now take the line that goes through $\rho$ and the center of the sphere: $\rho$ can also be expressed as the mixture of the states $+$ and $-$ which, since they represent equal and opposite directions, form a basis. To diagonalize $\rho$, then, means to express it in terms of $+$ and $-$. 

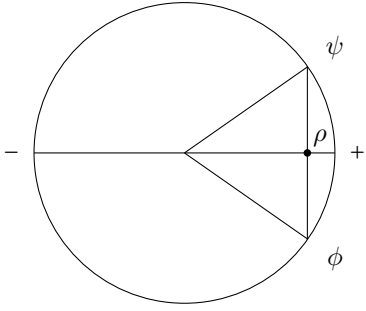
\begin{center}
	\begin{tikzpicture}[scale = 1]
		\draw (0,0) circle (2);
		\node at (-2.3,0) {$-$};
		\node at (2.3,0) {$+$};
		\node at (2,1.4) {$\psi$};
		\node at (2,-1.4) {$\phi$};
		\draw (-2,0) -- (2,0);
		\begin{scope}
			\clip(0,0) circle (2);
			\draw (0,0) -- (2,1.4);
			\draw (0,0) -- (2,-1.4);
			\draw (1.634,1.4) -- (1.634,-1.4);
		\end{scope}
		\fill (1.634,0) circle (0.05);
		\node at (1.8,.2) {$\rho$};
	\end{tikzpicture}
\end{center}

If $\theta_{\psi\phi}$ is the angle between $\psi$ and $\phi$ on the Bloch sphere, we have
\begin{equation}
	|\langle \psi | \phi \rangle |^2 = \cos^2 \frac{\theta_{\psi\phi}}{2}.
\end{equation}
The angle is divided by two because the angle on the Bloch sphere (i.e.~in physical space) is double the angle in the Hilbert space. For example, the angle between $+$ and $-$ on the Bloch sphere is $\pi$, but the inner product between them is zero (i.e.~opposite directions in physical space correspond to orthogonal states) and therefore the angle in the Hilbert space is $\pi/2$.

Now we express $\psi$ and $\phi$ in terms of $+$ and $-$, remembering that they form a basis. Given that $\rho$ is at the midpoint, the figure is vertically symmetric. The angle between $\psi$ and $+$, then, is half of $\theta_{\psi\phi}$. The inner product between $\psi$ and $+$ is
\begin{equation}
	\begin{aligned}
		|\langle \psi | + \rangle |^2 &= \cos^2 \frac{\theta_{\psi +}}{2} \\
		&= \cos^2 \frac{\theta_{\psi\phi}}{4}.
	\end{aligned}
\end{equation}
Keeping in mind that we are composing vectors in the Hilbert space (and not in the geometry of the physical space) we have
\begin{align*}
	\left|\psi\right>&=\cos\frac{\theta_{\psi\phi}}{4}\left|+\right>+\sin\frac{\theta_{\psi\phi}}{4}\left|-\right> \\
	\left|\phi\right>&=\cos\frac{\theta_{\psi\phi}}{4}\left|+\right>-\sin\frac{\theta_{\psi\phi}}{4}\left|-\right>.
\end{align*}

The density matrices corresponding to the pure states are
\begin{align*}
	\left|\psi\right>\left<\psi\right|&=\cos^2\frac{\theta_{\psi\phi}}{4}\left|+\right>\left<+\right|\\
	&+\cos\frac{\theta_{\psi\phi}}{4}\sin\frac{\theta_{\psi\phi}}{4}\left(\left|+\right>\left<-\right|+\left|-\right>\left<+\right|\right) \\
	&+\sin^2\frac{\theta_{\psi\phi}}{4}\left|-\right>\left<-\right| \\
	\left|\phi\right>\left<\phi\right|&=\cos^2\frac{\theta_{\psi\phi}}{4}\left|+\right>\left<+\right|\\
	&-\cos\frac{\theta_{\psi\phi}}{4}\sin\frac{\theta_{\psi\phi}}{4}\left(\left|+\right>\left<-\right|+\left|-\right>\left<+\right|\right) \\
	&+\sin^2\frac{\theta_{\psi\phi}}{4}\left|-\right>\left<-\right|.
\end{align*}

We can now calculate the mixture
\begin{align*}
	\frac{1}{2}(|\psi\rangle\langle\psi| &+ |\phi\rangle\langle\phi|) \\
	&=\cos^2\frac{\theta_{\psi\phi}}{4}\left|+\right>\left<+\right| +\sin^2\frac{\theta_{\psi\phi}}{4}\left|-\right>\left<-\right| \\
	&=\frac{1+\cos\frac{\theta_{\psi\phi}}{2}}{2}\left|+\right>\left<+\right| +\frac{1-\cos\frac{\theta_{\psi\phi}}{2}}{2}\left|-\right>\left<-\right| \\
	&=\frac{1+|\langle\psi|\phi\rangle|}{2}\left|+\right>\left<+\right| +\frac{1-|\langle\psi|\phi\rangle|}{2}\left|-\right>\left<-\right|. \\
\end{align*}

As $\rho$ is in a diagonal form, the entropy is given by
\begin{equation}\label{entropy}
	H(\rho) = H\left(\frac{1+|\langle\psi|\phi\rangle|}{2}, \frac{1-|\langle\psi|\phi\rangle|}{2}\right).
\end{equation}
If we define $p(\psi, \phi) = |\langle\psi|\phi\rangle|$ and $I(\psi, \phi) = H( \frac{1}{2}|\psi\>\<\psi| + \frac{1}{2}|\phi\>\<\phi|)$ we have
\begin{align}\label{entropy}
	I(\psi, \phi) &= f(p(\psi, \phi))
\end{align}
where
\begin{align}\label{entropy}
	f(p) & = - \frac{1+p}{2} \log \frac{1+p}{2} 
	- \frac{1-p}{2} \log \frac{1-p}{2}.
\end{align}
Note that $f$ restricted to the domain $[0,1]$, the domain on which probability is defined, is invertible. We have thus proved the proposition.
\end{proof}

Having shown the connection between superposition and non-unique decompositions of mixed states, and between the Born rule and the entropy of mixtures, we can now use these to justify the physicality of the complex inner product space structure for quantum states.
\begin{prop}
Complex inner product spaces are physical when used to represent the state space of quantum systems.
\end{prop}
\begin{justification}
Since we can physically prepare statistical ensembles, a physical theory must be able to characterize them. Part of that characterization includes keeping track of which particular ensemble is obtained by mixing particular states in particular proportions. A defining feature of quantum mixed states is that, unlike classical mixtures, they do not provide a unique decomposition in terms of pure states. This non-unique decomposition, as shown by Proposition \ref{prop_superpositionIsDecomposition}, is equivalent to the ability to create superpositions of pure states. Therefore, the inner product structure is the exact structure needed to characterize quantum ensembles and their properties, which is required by the physics. Note that the complexity of the inner product space influences how many decompositions are equivalent to the same mixture, and therefore is an expression of a physical requirement.

Additionally, a physical theory must be able to define the entropy of each ensemble, as required by thermodynamics. In quantum mechanics, the thermodynamic entropy is recovered by the von Neumann entropy, which Proposition \ref{prop_innerProductIsEntropy} shows is exactly captured by the inner product. Therefore, once again, the inner product structure is required to capture a necessary physical requirement.
\end{justification}

\section{The unphysicality of completeness}

To show that the completeness requirement on complex inner product spaces that represent quantum states is unphysical, we need to show that this definition forces the state space to include elements that cannot have physical correspondents.

The Hellinger–Toeplitz theorem states that a self-adjoint operator $O$ defined on the whole Hilbert space $\mathcal{H}$ is necessarily bounded. Therefore, an unbounded operator $X$ cannot be defined on the whole space, meaning that the operator must diverge or be ill-defined on a subset of $\mathcal{H}$. To see that the problem is the completeness requirement of the Hilbert space, we prove the following:
\begin{prop}
Let $\mathcal{H}$ be a Hilbert space and $X : D(X) \to Y$ an unbounded operator where $D(X)$ is the domain of $X$, and $Y$ is another Hilbert space. Then $D(X)$ is an inner product space that is not complete.
\end{prop}
\begin{proof}
The domain $D(X)$ of the operator is a subspace of $\mathcal{H}$. In fact, if $||X(v_1)||_Y$ and $||X(v_2)||_Y$ are finite, so is $||X(a_1 v_1+a_2v_2)||_Y$. Therefore $D(X)$ is an inner product space. It cannot be complete by the converse of the Hellinger–Toeplitz theorem.
\end{proof}

To understand the physical significance, consider the position operator $X : D(X) \to \mathcal{H}$. If we take a normalized vector $\psi \in \mathcal{H}$, the norm $||X\psi|| = \<X\psi|X\psi\> = E[X^2|\psi]$ is equal to the expectation of the position squared. The operator is unbounded because the position can be arbitrarily large. We can therefore construct a Cauchy sequence $\{\psi_i\}$ such that $\lim\limits_{i \to \infty}||X\psi|| = +\infty$, that is, we can take a sequence of vectors for which the expectation of position is larger and larger. Because of completeness, the limit of the Cauchy sequence will be in the Hilbert space. However, because the expectation diverges, the limit cannot be in the domain of $X$. Completeness forces us to include ``states'' for which the position is ill-defined.

The matter would still be manageable if we had a way to keep these degenerate cases isolated, but this is not the case. In general, a unitary operator can move vectors in and out of the domain.
\begin{prop}
Let $\mathcal{H}$ be a Hilbert space, $X : D(X) \to Y$ an unbounded operator, and $U : \mathcal{H} \to \mathcal{H}$ a unitary operator. Then, in general, $U(D(X)) \notin D(X)$.
\end{prop}
\begin{proof}
Let $\psi, \phi \in \mathcal{H}$ be two normalized vectors. Then we can construct a unitary operator $U$ such that $\phi = U(\psi)$. For example, we can take the operator that performs a rotation in the subspace identified by $\psi$ and $\phi$, leaving the orthogonal subspace unchanged. We can now take $\psi$ to be in $D(X)$ while $\phi$ to not be in $D(X)$. Therefore, in general, a unitary operator can move vectors in and out of the domain of an unbounded operator.
\end{proof}

To understand the physical significance, note that unitary operators are isomorphisms between Hilbert spaces. That is, mathematically, two Hilbert spaces are considered indistinguishable if there exists a unitary transformation between them. As we saw, unitary operations can change the domain in which unbounded operators are well defined: they can map states for which position is well defined onto states for which it is not. While the mathematical structure remains unchanged, the physical significance hasn't.

More concretely, a change of variable $x$ to $y$ induces a unitary operator as it is just a change of representation. The transformation rules for wave functions simply follow by the fact that $\psi^\dagger(x) \psi(x) = \rho(x)$ gives us a probability density over $x$. Therefore we must have
\begin{equation}
	\begin{aligned}
		\rho(x) dx &= \rho(y) dy \\
		\rho(y) &= \rho(x) \frac{dx}{dy} \\
		\psi^\dagger(y) \psi(y) &= \psi^\dagger(x) \psi(x) \frac{dx}{dy}.
	\end{aligned}
\end{equation}
Therefore we can set 
\begin{equation}
	\begin{aligned}
		\psi(y) &= \psi(x(y)) \sqrt{\frac{dx}{dy}}.
	\end{aligned}
\end{equation}
The above transformation is unitary.

\begin{prop}
	Let $\mathcal{H} = L^2(\mathbb{R})$ be the space of square integrable functions. Let $f : \mathbb{R} \to \mathbb{R}$ be an orientation-preserving diffeomorphism. Then the map $U : \mathcal{H} \to \mathcal{H}$ such that 
	$\phi(y) = U(\psi(x)) \mapsto \psi(f^{-1}(y)) \sqrt{\left|\frac{df}{dx}\right|^{-1}}$ is a unitary transformation.
\end{prop}
\begin{proof}
	Note that $f$ models a change of variables, therefore, in physics notation, $y = y(x) = f(x)$ and $x = x(y) = f^{-1}(y)$. Also note that both variables are real, therefore the Jacobian $\frac{df}{dx}$, which exists and  is nonzero because $f$ is a diffeomorphism, is real valued. Therefore $\left(\frac{df}{dx}\right)^\dagger = \frac{df}{dx}$. We have:
	\begin{equation}
		\begin{aligned}
			\int \phi^\dagger(y) \phi(y) dy &= \int U(\psi(x))^\dagger U(\psi(x)) dy \\
			 = \int \psi^\dagger(f^{-1}(y)) &\sqrt{\left|\frac{df}{dx}\right|^{-1}}\psi(f^{-1}(y)) \sqrt{\left|\frac{df}{dx}\right|^{-1}} dy \\
			 &= \int \psi^\dagger(f^{-1}(y)) \psi(f^{-1}(y)) \left|\frac{df}{dx}\right|^{-1} dy \\
			 &= \int \psi^\dagger(x(y)) \psi(x(y)) \frac{dx}{dy} dy \\
			 &= \int \psi^\dagger(x) \psi(x) dx.
		\end{aligned}
	\end{equation}
	The map $U$ preserved the inner product and therefore is a unitary transformation.
\end{proof}

The issue is that the position operator $X$ and the position operator $Y$ are different operators, and the domain of $X$ will be in general different from the domain of $Y$. Let $y=y(x)$ be our change of variable such that $y(0) = 0$. If $\psi(x)$ is a wave function expressed over $x$, and $\phi(y)$ represents the same state as a wave function over $y$, we have the following relationship:
\begin{equation}
\int_{0}^{y(x)} \phi^\dagger(y) \phi(y) dy = \int_{0}^{x} \psi^\dagger(x) \psi(x) dx
\end{equation}
To find a coordinate transformation that takes finite expectations to infinite expectations, we set the initial and final wave functions as follows:
\begin{align}
\psi(x) &= \sqrt{\frac{e^{-x^2}}{\sqrt{\pi}}} \\
\phi(y) &= \frac{1}{\sqrt{\pi(y^2 + 1)}}.
\end{align}
The first distribution is a Gaussian; therefore, all moments of the distribution, the expectation for $x^n$, are finite. The second goes to infinity like $\frac{1}{x^2}$, so the expectation will diverge for all moments above the first. We find
\begin{equation}
\begin{aligned}
\int_{0}^{y(x)} \frac{1}{\pi(\hat{y}^2 + 1)} d\hat{y} &= \int_{0}^{x} \frac{e^{-\hat{x}^2}}{\sqrt{\pi}} d\hat{x} \\
\frac{\tan^{-1}(y(x))}{\pi} &= \frac{\erf(x)}{2} \\
y(x) &= \tan \left(\frac{\pi}{2}\erf(x)\right). \\
\end{aligned}
\end{equation}
The graph of the above function increases so fast that, when plotted, it looks like it has a vertical asymptote. Nonetheless, it is an invertible function over the whole domain and will induce a unitary transformation on the Hilbert space.\footnote{Note that we could implement a similar coordinate transformation in classical mechanics, which would also change finite moments of a distribution into infinite ones. It may be that, like differentiable manifolds are constrained to only differentiable coordinate transformations, we need ``statistical manifolds'' that are constrained to coordinate transformations that preserve the finiteness of moments of distributions. This is a topic that requires further analysis.}

Unitary operators are also used in quantum theory to represent time evolution. Typically, we have a family of unitary operators $\{U_t\}_{t \in \mathbb{R}}$ that is strongly continuous (i.e. $\lim_{t \to t_0} U_t = U_{t_0}$) and respect composition in time (i.e. $U_t U_s = U_{t+s}$). We can turn the above transformation into an evolution by setting
\begin{equation}
z(x, t) = \cos(\omega t) x + \sin(\omega t) \tan \left(\frac{\pi}{2}\erf(x)\right).
\end{equation}
The idea is that the evolution will simply stretch and shrink the wave function intermittently, potentially oscillating expectation values from finite to infinite and vice versa in finite time. Clearly, this is not a physically meaningful time evolution, yet it will correspond to a strongly continuous one-parameter unitary group and, by Stone's theorem, will even admit a Hamiltonian.

The mathematician that points out that this is all known and there is nothing new is missing the point. The physicists that note that ``yes, those elements are clearly unphysical, but I don't use them, so I do not care,'' are also missing the point. The point is that, physically, we will want to make general statements such as ``for all observers, X is valid'' or ``under all evolutions, Y is valid.'' Given that Hilbert spaces mix clearly physically meaningful objects with clearly unphysical ones, we have no way to make those statements precisely.

We can conclude the discussion with the following:
\begin{prop}
Infinite-dimensional complex Hilbert spaces are unphysical when used to represent the state space of quantum systems due to completeness.
\end{prop}
\begin{justification}
If a finite-dimensional space models a quantum system, given that every finite-dimensional space is complete, finite-dimensional complex Hilbert spaces are physical.

Some physical systems are described by unbounded discrete quantities (e.g.~number of particles) or continuous quantities (e.g.~position), which require infinite-dimensional spaces. Infinite-dimensional complex inner product spaces, then, are needed to represent quantum systems.

An unbounded operator cannot be defined on all the elements of a Hilbert space by the converse of the Hellinger-Toeplitz theorem. This means that if one represents states as rays on a Hilbert space and observables as operators, the unbounded observable will not be well defined, not even as a statistical quantity, on all states. In particular, whether the statistical properties of an observable are well defined, or whether it is possible to identify a state through tomography experimentally, would be coordinate/representation dependent. Moreover, whether the expectation of an observable is finite or not would be something that can change in finite time.

On physical grounds, we must have that an observable has well-defined statistical properties over all states, that this property is covariant, and that it is preserved by time evolution. Hilbert spaces do not satisfy this requirement and are, therefore, unphysical.
\end{justification}

\section{The unphysicality of rigged Hilbert spaces}

As part of the feedback we received on previous drafts, some people argued that rigged Hilbert spaces were developed to solve exactly the problems we highlight. They were not. They are designed to solve another problem, namely extending the spectral theorem to the case of continuous spectra~\cite{Madrid_2005, reed_methods_1980, Landsman2017}.

First of all, note that ``rigged'' in this context does not mean ``manipulated by dishonest intention'', but rather ``fully equipped and ready to go'' as, for example, a fully-rigged ship. Therefore, the name already tells us that rigged Hilbert spaces are Hilbert spaces with additional structure. They still satisfy completeness, and thus are affected by the problems presented above. The unphysicality of a mathematical structure cannot simply be solved by adding more structure: the unphysical elements must be removed.

Rigged Hilbert spaces exist to address the fact that an operator may have eigenfunctions that are not in the Hilbert space. For example, the eigenstates for the operator $P = - \imath \hbar \partial_x$ are of the form $e^{\frac{xp}{-\imath\hbar}}$, and are not square integrable. This is due to the continuous spectrum, not the unboundedness of the operator, and is therefore an independent problem. We can show that with two examples.

Bounded operators can have a continuous spectrum. For example, let $\mathcal{H}$ be the Hilbert space of all wave functions defined on the compact domain $x \in [0,1]$. The position operator, in this case, is bounded. We do not have problems of infinite expectation value. However, the position operator $X$ does not have eigenfunctions in $\mathcal{H}$. The Dirac delta distribution, in fact, is not part of the Hilbert space; a sequence of Gaussians with decreasing variance, whose limit is a Dirac delta, is not a Cauchy sequence. Therefore we cannot properly talk about eigenstates of $X$ in $\mathcal{H}$.

Unbounded operators can have a discrete spectrum. For example, the harmonic oscillator has a discrete spectrum, with eigenfunctions given by the Hermite polynomials. All eigenfunctions are in the Hilbert space.

Let us give a very rough description of the rigged Hilbert space construction. Given a generic Hilbert space $\mathcal{H}$, we take an operator $X$, which we assume unbounded. The domain $D(X)$ of the operator is a strict subset of $\mathcal{H}$. That is, $D(X) \subset \mathcal{H}$. We can then define $D(X)^*$ as the space of linear functionals of $D(X)$.\footnote{A similar construction can be done with anti-linear functionals, giving the anti-dual space.} Mathematically, $D(X)$ is the space of test functions while $D(X)^*$ is the space of tempered distributions, the dual of $D(X)$. This is a strict superset of $\mathcal{H}$. That is,
\begin{equation}
	D(X) \subset \mathcal{H} \subset D(X)^*.
\end{equation}
This is called the Gelfand triple. This construction allows us to talk about $\delta(x)$, as it is a tempered distribution, and therefore it can be found in $D(X)^*$. Note that if we chose a different operator, say $P$, we would have a different construction and a different rigged Hilbert space.

The point of showing the construction is to see that the rigged Hilbert space is a Hilbert space plus two other spaces. Therefore, a rigged Hilbert space is physical only if all three spaces are physical. If the Hilbert space is not physical, neither is the rigged Hilbert space.

\section{The possible physicality of Schwartz spaces}

We have found that:
\begin{enumerate}
\item complex inner product spaces are needed
\item infinite-dimensional spaces are needed
\item completeness is unphysical.
\end{enumerate}
In other words, we need a stricter way to characterize convergence within the state space so that we can manage the infinite dimensionality. Conceptually, we need to treat the infinities coming from unbounded quantities as potential infinities, not actual infinities: we never actually have infinitely many particles or a system positioned at infinity; we have an infinite range of possible values.

A naive idea would be to require that the expectation value for all possible variables converges. Unfortunately, this doesn't work. Suppose we have a distribution $\rho(x)$ over position $x$. We would require the expectation of any function to be finite, which would include $e^x$, $e^{e^x}$, $e^{e^{e^x}}$, and so on. The only way this can work is if the function has compact support. Physically, this seems too restrictive as it excludes functions like normal distributions. At a closer investigation, the Paley-Wiener theorem~\cite{fourierCompact} tells us that the Fourier transform of a function with compact support is an entire function, and an entire function does not have compact support~\cite{markushevich2014entire}. Therefore, we cannot have compact support in both position and momentum.\footnote{This is not true for distributions in classical mechanics, as position and momentum are different variables, not linked by the Fourier transform. Therefore the space of distributions with compact support on phase space gives us finite expectation value for all functions of position and momentum.}

The only way out, then, is to require only a certain class of observables to have finite expectation. Physically, it would make sense to require exactly those observables needed to fully reconstruct the state through quantum tomography. A good set in that respect seems to be the set of polynomials of conjugate quantities, like position and momentum. In quantum mechanics, this means requiring all expectations of $X^nP^m = -\imath \hbar x^n\partial_x ^m$ to be finite. This restricts our space to those functions that are infinitely smooth (i.e.~all derivatives exist) and decrease very rapidly as $x$ goes to infinity: the Schwartz functions. Requiring finite expectation for all polynomials of position and momentum, then, does not give us the usual $L^2$ Hilbert space; it gives us the Schwartz space.

Not only do the expectations of all polynomials exist, but they are enough to identify the elements of the space and define limits. The Schwartz space is, in fact, a Fréchet space~\cite[Theorem V.9]{reed_methods_1980}, meaning that it is a Hausdorff space, with the topology induced by the seminorms given by the expectations, and it is complete with respect to those seminorms. Physically, this means that measurements on the polynomials (i.e.~the topology of the seminorms) are enough to identify the states (i.e.~the singletons, which are closed sets as per Hausdorff).\footnote{This essentially tells us that the moment problem restricted to pure states has a unique solution in quantum mechanics. Setting up a comparison with the classical case is delicate. If we restrict ourselves to classical pure states, i.e.~points in phase space, then the solution is also unique. If we restrict ourselves to compact support, distinct distributions have distinct moments even in the multi-variable case~\cite{moment_problem_2017}. If we restrict ourselves to integrable functions, even in the single variable case, distinct distributions can have the same moments. None of these comparisons seems adequate as a quantum state is a statistical object by nature, but it is not an arbitrary mixture. A more proper comparison would seem to restrict the problem to the space of integrable functions with zero entropy.} Moreover, if we have a sequence of states for which all expectations converge, we are guaranteed to converge to another state (i.e.~completeness with respect to the seminorms). If any expectation diverges, the limit is not in the space. This criterion of convergence seems a lot more physically sound.

To compare Schwartz and Hilbert spaces, the Schwartz space is a subspace of the $L^2$ Hilbert space taken as an inner product space. That is, all Schwartz functions are elements of the Hilbert space: they are closed under linear combinations and have the same inner product. It is not complete, and it is a dense subspace, meaning that any element of the Hilbert space can always be approximated by a Schwartz function with an arbitrary level of precision. Additionally, once a map is defined over the Schwartz space, the map is uniquely extended over the Hilbert space. It would then seem that all we are losing is the behavior at actual infinity.

Also note that Schwartz spaces over $\mathbb{R}^n$ are topologically different if $n$ is different. This means that Schwartz spaces can distinguish the number of the degrees of freedom at the topological level. This is in line with what happens in classical mechanics, where the phase space $\mathbb{R}^{2n}$ charted by position and momentum is topologically different for different values of $n$. Additionally, the completion of those Schwartz spaces for any $n$ is the same Hilbert space. This means that the completion actually loses physical information about the space, which again shows that the physical content is actually in the Schwartz space.

The Schwartz space also has another important property: it is closed under the Fourier transform. That is, the Fourier transform of a Schwartz function is a Schwartz function. Physically, this means that we can always express a wave function in either the position or the momentum basis. Mathematically, the theory of tempered distributions, which allows us to define objects like the Dirac delta distribution, is built on top of the Schwartz space. Therefore, this space already plays a fundamental role in the definitions of the tools we already use in physics.

If we restrict ourselves to the Schwartz space, then the only unitary transformations we are interested in are those that map Schwartz functions to Schwartz functions. These are special cases of unitary transformations over Hilbert spaces and inherit all their properties, such as the existence of the inverse. The change of coordinate $y(x)$ we defined in the previous section, then, is automatically ruled out. It is interesting to note that two Hilbert spaces connected by such unitary transformations have essentially two separate theories of distributions as they start from two different sets of objects. This shows exactly how two inequivalent physical representations can be mathematically equivalent.

Though we believe the arguments to be thoroughly convincing, we are stopping short of declaring Schwartz spaces physical. First, this is not the goal of the paper, which is to show the unphysicality of Hilbert spaces. Second, we would need to show the importance of expectation values of polynomials from experimental considerations.\footnote{The lack of a clear classical parallel for the moment problem restricted to pure states may make that argument difficult to construct.} It is not clear to us that that's the case, so it may be that the justification would have to proceed in a different, but ultimately equivalent, way.

Additionally, Schwartz spaces can only represent physical systems with a fixed number of degrees of freedom over Euclidean space. It is not clear how to generalize the definition to cover all physical cases.

Now, one may be concerned that if we abandon Hilbert spaces, we would also have to abandon many convenient tools built on top of them. This is not the case. When solving an integral, we sometimes extend the domain from real to complex values, which are unphysical, because complex analysis has nicer mathematical features. The result so found, however, is still a valid result in the real domain. Similarly, we can pose a problem on the Schwartz space, extend to the Hilbert space or the space of distributions for calculation, and then bring the result back to the Schwartz space. What needs to be rethought, instead, are manipulations of the state spaces themselves, especially under infinite operations. For example, if we want to construct the state space for a composite system of arbitrarily many subsystems which are characterized by arbitrarily large observables, understanding exactly which infinities are physical and which aren't is a key issue. Therefore taking the Hilbert space, which we saw doesn't deal with physical infinities correctly, and making infinitely many copies of it is bound to create problems.

\section{A final remark on the mathematical foundations of physical theories}

Once mathematicians and philosophers noticed that mathematical theories allowed one to express inconsistencies (e.g.~Russell's paradox, the Berry paradox, etc.), an effort was put into strengthening the foundations of mathematics so that such paradoxes were either resolved or couldn't be expressed. Concerning the mathematical foundations of physical theories, we have a similar problem: we can write statements that, though mathematically consistent, are physically inconsistent. Looking for and addressing these problems will ultimately lead to a more solid foundation, which, we believe, is the prerequisite to solving other outstanding issues in fundamental physics.

\end{document}